\documentclass[llncs]{llncs}
\usepackage[latin9]{inputenc}
\usepackage{color}
\usepackage{float}
\usepackage{amsmath}
\usepackage{amssymb}
\usepackage{graphicx}
\usepackage{comment}
\usepackage{a4wide}
\usepackage{tabularx}
\usepackage{enumitem}
\usepackage{tikz}
\usepackage{pgfplots}
\usepackage[colorlinks]{hyperref}
\usepackage{tablefootnote}
\usepackage{footnote}

\makesavenoteenv{tabular}
\makesavenoteenv{table}
\makeatletter

\floatstyle{ruled}
\newfloat{algorithm}{tbp}{loa}
\providecommand{\algorithmname}{Algorithm}
\floatname{algorithm}{\protect\algorithmname}

\newcommand{\suppress}[1]{}

\usepackage{tikz}
\usepackage{aliascnt}
\usepackage{tensor}
\newcommand\Perms[2]{\tensor[^{#2}]P{_{#1}}}

\def \g {\gamma}

\def \F {{\mathbb F}}
\def \Z {{\mathbb Z}}

\def \deg {{\rm deg}}

\def \Ga{{\alpha}}

\def \fR {\mathfrak{R}}
\def \fQ {\mathfrak{Q}}
\def \mL {\mathcal{L}}
\def\ba{{\bf a}}

\def\bv{{\bf v}}

\makeatother

\begin{document}

\title{On the Closest Vector Problem for Lattices Constructed from Polynomials and Their Cryptographic Applications}

\author{Zhe Li\inst{1}, San Ling\inst{1}, Chaoping Xing\inst{1}, Sze Ling Yeo\inst{2}}
\institute{School of Physical and Mathematical Sciences, Nanyang Technological University
\and Institute for Infocomm Research (I2R), Singapore}
\maketitle

\begin{abstract}
In this paper, we propose new classes of trapdoor functions to solve the closest vector problem in lattices. Specifically, we construct lattices based on properties of polynomials for which the closest vector problem is hard to solve unless some trapdoor information is revealed. We thoroughly analyze the security of our proposed functions using state-of-the-art attacks and results on lattice reductions. Finally, we describe how our functions can be used to design quantum-safe encryption schemes with reasonable public key sizes. In particular, our scheme can offer around $106$ bits of security with a public key size of around $6.4$ \texttt{KB}. Our encryption schemes are efficient with respect to key generation, encryption and decryption.
\end{abstract}

\section{Introduction}

In today's digital world, protecting the confidentiality and integrity of digital information is of vital importance. At the core of providing data privacy, integrity and authenticity are a class of algorithms called public-key cryptosystems, first introduced by Diffie and Hellman in 1976 \cite{DH76}. Essentially, these public-key cryptosystems are constructed from trapdoor functions. Recall that a trapdoor function $f$ is a function satisfying:
\begin{itemize}
\item $f(x)$ is easy to evaluate for all inputs $x$;
\item Given an output $y$ of the function $f$, it is computationally infeasible to determine $x$ such that $y=f(x)$ unless some trapdoor information is known.
\end{itemize}
To date, most commonly deployed trapdoor functions rely on some computational number theory problems where no efficient classical algorithm is known, including the integer factorization problem and discrete logarithm problem in various finite groups. However, Shor showed in 1994 that there exists a quantum algorithm that can solve these problems in polynomial time \cite{S94}.

As such, there is an urgent need to design new trapdoor functions based on different mathematical problems that are resistant to quantum algorithms. At present, a number of potential classes of mathematical problems are being considered and studied, namely, from coding theory, lattices, multi-variate polynomials, hash functions and isogenies of supersingular elliptic curves \cite{Ber08}. Among them, lattices seem to be among the most promising, spawning many new constructions with different properties and capabilities, most notably fully homomorphic encryption \cite{G09}.
\subsection{Previous work}
Early lattice-based encryption schemes include the Ajtai-Dwork Encryption \cite{AjtaiD97}, Goldreich-Goldwasser-Halevi (GGH) encryption \cite{GGH97} and NTRU encryption \cite{Hof98}. A breakthrough of modern lattice-based cryptography is the invention of the learning with error (LWE) and Ring-LWE problems \cite{Regev09,LPR13}. Consequently several LWE-based encryption schemes have been proposed \cite{BV11,BV14,GHV10,GVW15}.

The GGH encryption scheme is an analog of the famous coding-based encryption scheme--McEliece encryption.
In the original paper \cite{GGH97}, 5 different challenges for lattice dimensions ranging from $n=200$ to $n=400$ were proposed. Unfortunately, all the challenges, except for the instance with $n=400$, were broken. Indeed, it was shown in \cite{Ngu99} that the structure of the error provided an inherent weakness and together with the embedding technique, this weakness can be exploited to attack the GGH instances. Even though suggestions were put forth in \cite{Ngu99} to mend the scheme, the corresponding parameters will make the scheme impractical for use.
\subsection{Our work}
In this paper, we seek to design new trapdoor functions in which the function inversion involves solving the closest vector problem (CVP), one of the well-known hard lattice problems. Our construction is primarily inspired by the GGH construction \cite{GGH97} and the McEliece code-based cryptosystem \cite{M78}. Like these schemes, our function involves constructing a point that is sufficiently close to a certain point in a lattice determined by the input. Hence, inverting this function will require one to solve the closest vector problem (CVP).

More precisely, we choose $n$ distinct elements $\Ga_1,\Ga_2,\dots,\Ga_n$ of $\F_q$ and $t$ distinct monic irreducible polynomials $c_1(x),c_2(x),\dots,c_t(x)$ of degree $d_0$ such that ~$\gcd ( \prod_{i=1}^n(x-\Ga_i), \prod_{i=1}^tc_i(x) ) = 1$~. An integer point $(a_1,a_2,\dots,a_n)\in \mathbb{Z}^n$ is a lattice point if and only if $\prod_{i=1}^n(x-\Ga_i)^{a_i}\equiv1\pmod{c(x)}$, where $c(x)=\prod_{i=1}^tc_i(x)$. Then a basis of this lattice can be computed efficiently. We show in this paper that, given $q,n$ and certain range of $d=td_0$, and the basis of this lattice, the embedding technique does not work well to tackle the CVP of this lattice. On the other hand, with information on $\{\Ga_1,\Ga_2,\dots,\Ga_n\}$ and the polynomial set $\{c_1(x),c_2(x),\dots,c_t(x)\}$, we are able to efficiently solve the CVP of this lattice. With this trapdoor, we can design an encryption algorithm that is similar to the one in the GGH encryption scheme.

Furthermore, we conduct a thorough security analysis on our trapdoor function. We show that we can design an encryption scheme based on our trapdoor function that is resistant against existing attacks with public key sizes smaller than those proposed in \cite{GGH97}. A practical encryption based on this trapdoor is also provided.

For a given security parameter $\lambda$, choose $n$ and $d$ satisfying $n \geq 200, d < n/2, \sqrt{\frac{n}{2\pi e(d-1)}} (n+2d)^{d/n} \le 0.3*1.007^n, \lambda \leq 36.4\log_2{n}$, $q=\texttt{next\_prime}(n, d)$ and ${n -d \choose l} \ge 2^\lambda$, where $l = \frac{(n-d)(d-1)}{n}$. For a concrete parameter selection, refer to Section \ref{sec:parameter_section}.

\subsection{Comparison}
Although the encryption algorithm in our scheme is similar to the one in the GGH encryption scheme, the trapdoor function in our scheme is totally different from that in the GGH encryption scheme. For the GGH scheme, one first chooses a ``nice" basis of a lattice so that solving CVP is easy, and then multiplies with a unimodular matrix and permutation matrix for confusion so that solving CVP is no longer easy.
However, unlike the GGH cryptosystem, we do not rely on constructing a good basis as a trapdoor. Instead, by using lattices constructed from polynomial functions, one can invert the function efficiently as long as one has access to the polynomials and points involved. As such, the vector norms of our basis can be better controlled, one of the main limitations of the GGH scheme.

In \cite{Ngu99}, attacks were proposed which essentially rendered the GGH cryptosystem insecure for practical parameters (it is still not broken asymptotically). However, our experiments show that the trapdoor in this paper is resistant to the existing attacks including the attack given in \cite{Ngu99}. Apart from security advantage, the public key size of our encryption scheme is smaller compared with the GGH encryption scheme. The encryption and decryption complexity is almost the same as for the GGH encryption scheme. The comparison between the GGH scheme and our polynomial lattice scheme is given in Table \ref{tab:1}.

Note that $n$ in Table \ref{tab:1} represents the rank of lattices. For time complexity computation, we assume that multiplication of two integers less than $n$ requires $O(\log n\log\log n)$ bit operations and multiplication of two degree-$d$ polynomials over $\F_q$ with $d\le q$ requires $O(d\log d)$ field element multiplications.
\newpage
\begin{table}[h!]
\centering
\caption{Comparison with GGH }
\begin{tabular}{|c|c|c|}
\hline
&GGH &Polynomial lattice scheme\\ \hline
Public key size    &$n^2$ & ${n^2}/5\sim{n^2}/4$  \\ \hline
Encryption time          &$O(n^2 \log n \log \log n)$&$O(n^2 \log n \log \log n)$\\ \hline
Decryption time        &$O(n^2 \log n \log \log n)$&$O(n^2 (\log n)^2 \log \log n)$  \\ \hline
Resistant to embedding attack& No& Yes \\ \hline
Public key entry bit & $0.3n \tablefootnote{For each $n \in \{100, \ldots, 400\}$, we computed the mean value of the entry of GGH public key by repeating experiments $100$ times. Then applying linear congression tool, we get public key entry bit of GGH approximately equal to $0.3n$.}$ & $\log_2{(2n)}$ \\ \hline
\end{tabular}
\label{tab:1}
\end{table}

\subsection{Organization of the paper}

This paper is organized as follows. In the next section, we briefly summarize some important background on lattices as well as the two encryption schemes (namely, GGH and McEliece schemes) that inspire our work. In Section 3, we describe a family of lattices constructed from polynomials. We then present our new trapdoor functions based on these lattices in Section 4. This is followed by a security analysis on these trapdoor functions in Section 5. In Section 6, we give details of a semantically secure encryption scheme based on our trapdoor functions. In addition, we propose some possible parameters for our scheme. 

\section{Preliminaries}

\subsection{Background on Lattices} \label{sec:background}

In this section, we briefly review some of the important definitions and notions on lattices. We refer the reader to \cite{MG02,Ber08} for more background materials.

Let $n$ be a positive integer. By usual convention, we will represent vectors in $\mathbb{R}^n$ in the row form. For $\vec{x} = (x_1, \ldots, x_n) \in \mathbb{R}^n$, denote by $\vert\vec{x}\vert$ the Euclidean norm of $\vec{x}$, that is, $\vert\vec{x}\vert = \sqrt{\sum_{i=1}^nx_i^2}$.

{\bf Lattice:}
A lattice $L$ is a discrete additive subgroup of $\mathbb{R}^n$. Concretely, for $m \le n$, let $\{\vec{b}_1, \ldots, \vec{b}_m\}$ be $m$ linearly independent vectors in $\mathbb{R}^n$. Then a lattice $L$ is a set $\{a_1\vec{b}_1 + a_2\vec{b}_2 + \ldots +a_m\vec{b}_m: a_i \in \mathbb{Z}, i = 1, 2, \ldots, m\}$. $m$ is called the {\em dimension} or {\em rank} of the lattice. If $m=n$, then $L$ is said to have full rank. In this work, we will only focus on full-rank lattices. Further, $B = \{\vec{b}_1, \ldots, \vec{b}_m\}$ is called a {\em basis} of $L$. Let $M$ be the $m$ by $n$ matrix with rows $\vec{b}_i$, $i=1, 2, \ldots, m$ . Then, the {\em determinant} of $L$ (or the volume of $L$) is given by ${\rm det}(L) = {\rm vol}(L) = \sqrt{\vert MM^T \vert}$.

{\bf $n$-Ball:}
For $r \in \mathbb{R}$, let $B_n(r) = \{\vec{x} \in \mathbb{R}^n: \vert \vec{x}\vert \le r\}$ denote the $n$-dimensional ball centered around the origin with radius $r$. The volume of $B_n(r)$ is given by $V_n(r) = \frac{\pi^{n/2}r^n}{\Gamma(n/2+1)}$, where $\Gamma(\cdot)$ is the Gamma function.

{\bf Short vectors of a lattice:}
As a lattice is a discrete subgroup of $\mathbb{R}^n$, the set of all their Euclidean norms forms a discrete subgroup of $\mathbb{R}$. Hence, each lattice $L$ has a nonzero point such that its norm is the minimum. We denote this minimum norm by $\lambda_1(L)$. More generally, for $i=1, 2, \ldots, n$, $\lambda_i(L)$ denotes the smallest radius $r$ such that the ball $B_n(r)$ contains $i$ linearly independent points in $L$.

{\bf Gaussian heuristic:}
The Gaussian Heuristic estimates the number of lattice points in certain sets. Let $L$ and $S$ be a full-rank lattice and a connected $n$-dimensional object, respectively. Then the number of lattice points in $S$ is approximated by ${\rm vol}(S)/{\rm det}(L)$. This leads to the following Gaussian heuristic estimate on the shortest vector $\lambda_1(L)$ for a random lattice  $L$:
$\lambda_1(L) \approx V_n(1)^{-1/n}{\rm det}(L)^{1/n} \approx \sqrt{n/2\pi e}{\rm det}(L)^{1/n}$.

{\bf Lattice reduction:}
Any lattice has an infinite number of bases. In particular, given a basis of a lattice $L$, one can construct a new basis by multiplying the matrix formed by the basis vectors with unimodular integer matrices, that is, integer matrices with determinant $\pm 1$. In general, one often looks for a basis with short vectors or nearly orthogonal vectors. There are various algorithms to reduce a basis of a lattice into a basis of better quality. Well-known reduction algorithms include the LLL algorithm \cite{LLL82} and the BKZ algorithm \cite{Sch94}. In the BKZ algorithm, one essentially tries to find short vectors in the sub-lattice formed by sub-blocks of basis vectors. In fact, the LLL algorithm can be viewed as a special case of the BKZ algorithm where we work with pairs of vectors each time. Evidently, a BKZ algorithm with a bigger block size produces a basis with shorter vectors but this is achieved at the expense of a longer running time. Recent efficient implementations of the BKZ algorithm include the BKZ2.0 algorithm \cite{CN11} (where pruning was used to find the shortest vector for each sub-lattice) and the progressive BKZ algorithm \cite{Aono16} (where block sizes are progressively increased). When attempting to solve some lattice problems, one typically reduces the given basis using an appropriate lattice reduction algorithm before applying other algorithms.

{\bf Hermite factor:}
Let $\vec{b}_1$ denote the shortest vector in a given basis of a lattice $L$. To measure the quality of the basis, one often looks at the Hermite factor. The {\em Hermite factor}, denoted by $\delta^n$, is defined as $\delta^n = \frac{\vert\vec{b}_1\vert}{{\rm det}(L)^{1/n}}$. Typically, a smaller Hermite factor implies a better basis with a shorter vector. One may also simply consider the root Hermite factor $\delta$. For a random input basis, it was experimentally shown that the value of $\delta$ is determined by the particular reduction algorithm used \cite{GN08} and independent of the dimension of the lattice. In other words, one may use the root Hermite factor as a measure of the effectiveness of a reduction algorithm on random lattices.

{\bf Shortest Vector Problem (SVP):}
Given a lattice $L$, the {\em shortest vector problem} (SVP) seeks a nonzero point $\vec{v}$ in $L$ such that $\vert \vec{v}\vert = \lambda_1(L)$. For low dimensions, some proposed approaches to solve SVP include computing the Voronoi cell of the lattice and sieving (see \cite{HPS11} for details) as well as enumeration methods \cite{Ka83,GNR10}. However, these methods have complexity at best exponential in the lattice dimension $n$. As such, it becomes computationally infeasible to solve SVP when $n$ is large, say greater than $100$. Variants of the SVP have been proposed and extensively employed in lattice-based cryptography. Some of the main ones include:

\begin{itemize}
\item {\em Hermite SVP}: Find $\vec{v} \in L$ such that $\vert\vec{v}\vert \le \alpha {\rm det}(L)^{1/n}$ for some $\alpha$;
 \item {\em Unique SVP}: Given that $\lambda_2(L)/\lambda_1(L) \ge \gamma$, solve SVP in $L$.
\end{itemize}

In fact, the LLL algorithm \cite{LLL82} solves the Hermite SVP in polynomial-time for $\alpha$ exponential in $n$. It has been experimentally shown in \cite{GN08} that $\alpha \approx \delta^n$ where the root Hermite factor $\delta \approx 1.0219$. Similarly, the expected root Hermite factor and the corresponding time complexities for the BKZ algorithm with different block sizes were reported in \cite{GN08,CN11}.

In \cite{GN08}, Unique SVP was also studied. The authors demonstrated that if $\gamma = \epsilon\delta^n$ for some small constant $\epsilon$, then Unique SVP can be solved. It follows that when the ratio of the norms of the two shortest vectors in $L$ is at least $\epsilon\delta^n$, using the reduction algorithm with root Hermite factor $\delta$ will likely recover the shortest vector of $L$.

{\bf Closest Vector Problem/Bounded Distance Decoding (CVP/BDD): }
A closely related problem to SVP is the {\em closest vector problem} (CVP) or the {\em bounded distance decoding problem} (BDD). Given a target vector $\vec{t} \in \mathbb{R}^n$, CVP seeks a vector $\vec{v} \in L$ that minimizes $\vert\vec{t}-\vec{v}\vert$. In \cite{GMSS99}, the embedding technique was proposed to convert the CVP into SVP. Essentially, suppose that we have a target vector $\vec{t}$ that is close to a lattice generated by a basis $B = \{\vec{b}_1, \ldots, \vec{b}_n\}$. We construct another lattice $L' \subset \Z^{n+1}$ generated by the following matrix:
$$B'=\left(\begin{matrix}\vec{b}_1 & 0\\
\vec{b}_2&0\\
\vdots & \vdots\\
\vec{b}_n&0\\
\vec{t}&1
\end{matrix}\right).$$

It is easy to check that $L'$ has shortest norm given by $\lambda_1(L')^2 = \vert \vec{t}-\vec{v}\vert^2+1$, where $\vec{v}$ is a point in $L$ closest to $\vec{t}$. Hence, if $\vec{t}$ is close enough to $L$, this gives a corresponding short vector in $L'$.

A more direct approach to solve CVP is via Babai's nearest plane algorithm \cite{Bab85}. Let $B = \{\vec{b}_1, \ldots, \vec{b}_n\}$ be a reduced basis of $L$. Write $\vec{v} = v_1\vec{b}_1 + \ldots + v_m\vec{b}_m$. Essentially, Babai's nearest plane algorithm  progressively finds $v_m, v_{m-1}, \ldots$ that minimizes the projected vector $\pi_k(\vec{v}-\vec{t})$ as $k$ goes from $n$ to $k=1$. It can be shown that Babai's algorithm can output the correct result when the error $\vec{e} = \vec{t}-\vec{v}$ satisfies $\vec{e} \in P_{1/2}(B)$, where $P_{1/2}(B) = \{a_1\vec{b}_1 +\ldots +a_n\vec{b}_n: -1/2 \le a_i \le 1/2 \textrm{ for all } i = 1, \ldots, n\}$. Alternatively, one must have $\vert \langle\vec{e}, \vec{b}_i^*\rangle\vert \le \vert\vec{b}_i^*\vert^2/2$, where $\vert \langle\vec{e}, \vec{b}_i^*\rangle\vert$ is the absolute value of inner product of $\vec{e}$ and $\vec{b}_i^*$ and $\{\vec{b}_1^*, \ldots, \vec{b}_n^*\}$ is the Gram-Schmidt orthogonalized basis of $B$. In particular, the error norm must be relatively small with respect to the norms of the orthogonalized basis vectors.

Other improved methods to solve CVP/BDD include Lindner and Peikert's generalization of the Babai's nearest plane algorithm \cite{LP11} as well as enumeration approaches \cite{Ka83,Sch94,GNR10}. Lattice enumeration was proposed in \cite{Ka83} to solve SVP but can be easily adapted to solve BDD. Suppose that we know that $\vert \vec{v} - \vec{t}\vert \le R$ for some enumeration radius $R$. The enumeration algorithm is a generalization of Babai's nearest plane algorithm in the following sense. As the level number $k$ goes from $m$ down to $1$, one finds all the integers $v_k$ such that $\vert \pi_k(\vec{v}-\vec{t})\vert \le R$. In this way, we construct an enumeration tree where the leaves are the error vectors $\vec{v}-\vec{t}$ and the parent of each node at level $k$ is its projection onto $\pi_{k+1}(L)$.

In \cite{Sch94}, Schnorr and Euchner proposed pruned enumeration to speed up lattice enumeration but at the expense of reducing the probability of success. Briefly, instead of enumerating at each level using the fixed enumeration radius $R$, one constructs $m$ pruning coefficients $(P_1, P_2, \ldots, P_m)$ so that at each level $k$, the levelled enumeration radius is reduced to $R_k^2 = P_kR^2$. Here, the $P_k$'s must satisfy $0 < P_m \le P_{m-1} \le \ldots \le P_1 =1$. Hence, at each level $k$, the nodes are constructed as the coefficients $v_k$ resulting in $\vert \pi_k(\vec{v}-\vec{t})\vert \le R_k$.

Since the leveled enumeration radius is reduced, the probability of success is correspondingly decreased. This can be overcome by performing the algorithm multiple times with different input bases of the given lattice. Hence, there is a need to find a good balance between the time to perform the basis reduction and the time to perform the enumeration algorithm. {\em Extreme pruning } was suggested in \cite{GNR10}, where the pruning coefficients are chosen to result in a small probability of success. The authors argued that while the probability of success is much reduced, the reduction in the time to perform the enumeration algorithm is even greater, thereby decreasing the overall time. Finally, in \cite{Aono14}, the authors proposed an explicit algorithm to compute optimized coefficients given a fixed reduction algorithm (and hence, a certain root Hermite constant) and a desired probability of success. The same assumptions as used in \cite{GNR10} were used in \cite{Aono14}. Based on their algorithms, tables of optimized pruning coefficients were provided for some parameters.

We remark that solving BDD is one of the approaches to solve the learning with errors (LWE) problem. As such, various experiments had been performed on LWE instances via the BDD approach \cite{LN13,Kirshanova:PIO}.
A good discussion of the various approaches to solve CVP/BDD can be found in \cite{APS15}.

\begin{remark}
We remark that when the basis $B$ of a lattice $L$ is orthogonal, then solving the SVP and CVP become easy. This motivates us to define the {\em orthogonality defect} of a basis. Let $B=\{\vec{b}_1, \ldots, \vec{b}_n\}$ be a basis of a lattice $L$. Then the orthogonality defect of $B$ is defined by
$$h(B) = \frac{\prod_{i=1}^n\vert \vec{b}_i\vert}{{\rm det}(L)}.$$

Observe that $h(B) \ge 1$, and $h(B)=1$ whenever $B$ is orthogonal. It is suggested in \cite{GGH97} that Babai's nearest plane algorithm solves the CVP with respect to $B$ when $h(B)$ is close to $1$.
\end{remark}

\subsection{The McEliece Cryptosystem}

The McEliece encryption scheme was proposed in \cite{M78} as a public-key cryptosystem that is based on hard problems in algebraic coding theory instead of the usual integer factoring problem or discrete logarithm problems in groups. More specifically, its construction  hinges on the difficulty to decode general linear codes  over finite fields. Unlike the widely-used integer factoring problem and discrete logarithm problem which had been proven to be vulnerable to polynomial-time quantum algorithms \cite{S94}, the decoding problem is touted as one of the potential candidates to be used as a basis for post-quantum cryptography.

Essentially, the McEliece encryption scheme generates two linear codes, one with an easy decoding strategy while the other is presumably difficult to decode. Concretely, let $G$ be an $[n,k,2t+1]$ binary Goppa code which admits an efficient decoding algorithm. Let $U$ be a $k \times k$ invertible binary matrix and $P$ an $n \times n$ permutation matrix. Let $G' = UGP$. Then $G'$ represents a general linear code with no obvious way to decode. The basic structure of the McEliece encryption scheme can be described as follows:

{\bf Public key: }
The matrix $G'$ and the parameters $n,k,t$.

{\bf Private key: }
The matrices $G,U$ and $P$.

{\bf Encryption: }
Let $\vec{m}$ be a $k$-bit message. Randomly pick an $n$-bit error vector $\vec{e}$ with Hamming weight $t$. The encryption of $\vec{m}$ is given by
$$\vec{c}=E(\vec{m}) = \vec{m}G'+\vec{e}.$$

{\bf Decryption:}
Let $\vec{c'} = \vec{c}P^{-1}$. With the secret key $G$, decode $\vec{c'}$ to obtain the message $\vec{m'}$. Compute $\vec{m} = \vec{m'}U^{-1}$.

At present, the most effective attacks on the McEliece cryptosystem are variants of the information-set decoding attack \cite{LB88,S88}. In \cite{B08}, the authors successfully attacked the parameters proposed in the original paper. Nonetheless, these attacks remain exponential in the parameters and the security can be improved by increasing the parameter sizes. The main disadvantage of the McEliece cryptosystem is therefore, the relatively large public key sizes. For example, it was suggested that public key sizes of $256$ \texttt{KB} and $512$ \texttt{KB} are needed to ensure around $146$-bit and $187$-bit security, respectively.

\subsection{The GGH Cryptosystem}

The GGH cryptosystem was presented in \cite{GGH97} as a lattice analog of the McEliece cryptosystem. While the McEliece scheme exploits the difficulty to decode the received word obtained from a random code whenever errors of small weights are introduced, the GGH scheme relies on a similar phenomenon on general lattices. Indeed, the GGH scheme constructs two different bases of the same lattice, one of which allows CVP to be solved efficiently via Babai's nearest plane algorithm while the other basis is constructed as a random basis of the lattice and hence, has very poor performance with respect to CVP. More precisely, one first constructs a basis $B$ with short highly orthogonal rows (that is a basis with small orthogonality defect) and then multiply $B$ by random unimodular matrices to obtain a basis $B'$ of the same lattice with much higher orthogonality defect. $B$ is then used as the private key while $B'$ will serve as the public key. The basic structure of the GGH encryption scheme is presented next.

{\bf Private key:}
The basis $B$ and a unimodular matrix $U$;

{\bf Public key:}
The basis $B'=UB$, and the parameter $n$ and a small positive integer $\sigma$;

{\bf Encryption:}
Let the message $\vec{m} \in \mathbb{Z}^n$. Choose an error $\vec{e} \in \mathbb{Z}^n$ whose entries are randomly picked to be $\pm \sigma$. The encryption of $\vec{m}$ is given by
$$\vec{c} = E(\vec{m}) = \vec{m}B'+\vec{e}.$$

{\bf Decryption: }
Using Babai's nearest plane algorithm and the basis $B$, determine the vector $\vec{v}$ closest to $\vec{c}$. We have $\vec{m} = \vec{v}U^{-1}$.

Observe that, in order for the GGH scheme to work, one must be able to solve CVP with $B$ but not with $B'$. It follows that one needs to multiply $B$ by suitably dense unimodular matrices. As a result, the entries in $B'$ tend to be much larger.

Although the GGH encryption scheme is not resistant to the embedding algorithms for $n<400$, the scheme is still asymptotically secure. Furthermore,  the underlying ideas of the GGH construction remain interesting, particularly serving as trapdoor functions where inverting the function amounts to solving the BDD problem.

\section{Polynomial Lattices}
In this section, we give a new construction of lattices via polynomials over a finite field.
Let $q$ be a prime power. We denote by $\F_q$ the finite field with $q$ elements. Let $\fR$ denote the polynomial ring $\F_q[x]$. Fix a monic polynomial $c(x) \in \fR$ of degree $d$ and let $\fQ_{c(x)}$ denote the quotient ring $\fR/c(x)$.  Let $\fQ_{c(x)}^*$ denote the unit group of $\fQ_{c(x)}$, i.e., let $\fQ_{c(x)}^*=\{\overline{f(x)}\in \fQ_{c(x)}:\; \gcd(f(x),c(x))=1\}$. It is easy to verify that $\fQ_{c(x)}^*$ forms a multiplicative group. Furthermore, the cardinality of $\fQ_{c(x)}^*$, denoted by $\Phi(c(x))$, is given by the following formula.

\begin{lemma}
\label{unitgroup_lem}\cite[Lemma 3.69]{lidl1997finite}
Let $c(x)$ have the canonical factorization $\prod_{i=1}^tc_i(x)^{e_i}$, where $c_i(x)$'s are pairwise distinct monic irreducible polynomials over $\F_q$, $d_i$ are the degrees of $c_i(x)$ and $e_i \ge 1$. We have
$$\Phi(c(x)) = \prod_{i=1}^t(q^{e_i n_i}-q^{(e_i-1) n_i}).$$
\end{lemma}

Let $\alpha_1, \ldots, \alpha_n$ be $n$ distinct elements in $\F_q$ such that $c(\alpha_i) \ne 0$ for $i=1, \ldots, n$. Denote by $\ba$ the vector $(\alpha_1, \ldots, \alpha_n)\in\F_q^n$.

Define the map
\begin{alignat*}{10}
     \phi_\ba: &   \qquad \Z^n      & \longrightarrow & \qquad\qquad \fR                    &\longrightarrow & \qquad \fQ^*_{c(x)} \\
 & (u_1, \ldots, u_n) & \longmapsto     & f = \prod_{i=1}^n(x-\alpha_i)^{u_i} &\longmapsto     & f(x) \mod c(x).
\end{alignat*}

Observe that $\phi_\ba$ is a group homomorphism from $\Z^n$ to $\fQ^*_{c(x)}$. Let $\mL_{\ba,c(x)}$ denote the kernel of $\phi_\ba$. As $\mL_{\ba,c(x)}$ is a subgroup of $\Z^n$, $\mL_{\ba,c(x)}$ is a lattice. The following lemma provides some important properties of $\mL_{\ba,c(x)}$.

\begin{lemma}
\label{lattice_properties_lem}
The lattice $\mL_{\ba,c(x)}$ defined above satisfies the following properties:

\begin{itemize}
\item[{\rm (i)}]  $\mL_{\ba,c(x)}$ has rank $n$.
\item[{\rm (ii)}]  The determinant $\det(\mL_{\ba,c(x)})$ is upper bounded by $\Phi(c(x))$. Furthermore, $\det(\mL_{\ba,c(x)})=\Phi(c(x))$ if $\phi_\ba$ is surjective.
\item[{\rm (iii)}]  $\lambda_1(\mL_{\ba,c(x)}) \ge \sqrt{d}$. Moreover, $\lambda_1(\mL_{\ba,c(x)}) = \sqrt{d}$ if and only if there exists a lattice point with $d$ nonzero entries which are either all $1$ or $-1$.
\end{itemize}
\end{lemma}

\begin{proof}
\begin{itemize}
\item[{\rm (i)}]  Observe that for each $i=1, 2, \ldots, n$, we have $(0,0,\Phi(c(x)),\ldots,0) \mapsto (x-\alpha_i)^{\Phi(c(x))} \mapsto 1$ under $\phi_\ba$. Hence, each of these points is in $\mL_{\ba,c(x)}$. As these $n$ points are clearly linearly independent, they form a sub-lattice of $\mL_{\ba,c(x)}$ of rank $n$. Consequently, $\mL_{\ba,c(x)}$ has rank $n$.

\item[{\rm (ii)}]  As $\Z^n/\mL_{\ba,c(x)} \simeq{\rm Im}(\phi_\ba)\le \fQ^*_{c(x)}$
and  $\det(\mL_{\ba,c(x)}) = [\Z^n: \mL_{\ba,c(x)}]\det(\Z^n) = [\Z^n: \mL_{\ba,c(x)}]=|{\rm Im}(\phi_\ba)|$, we obtain the desired inequality. In addition, if $\phi_\ba$ is surjective, then $\Z^n/\mL_{\ba,c(x)} \simeq \fQ^*_{c(x)}$. Hence the equality follows.

\item[{\rm (iii)}]   Let $\bv = (v_1,v_2,\ldots,v_n)$ be a nonzero point in $\mL_{\ba,c(x)}$. Denote by $I$ and $J$ the sets $\{1\leq i\leq n:\; v_i>0\}$ and $\{1\leq j\leq n:\; v_j<0\}$, respectively. By definition of $\mL_{\ba,c(x)}$, we have $\prod_{i\in I}(x-\Ga_i)^{v_i}\prod_{j\in J}(x-\Ga_j)^{v_j}-1\equiv 0\mod c(x)$, i.e., the nonzero polynomial $\prod_{i\in I}(x-\Ga_i)^{v_i}-\prod_{j\in J}(x-\Ga_j)^{-v_j}$ is divisible by $c(x)$. Hence, $\deg(\prod_{i\in I}(x-\Ga_i)^{v_i}-\prod_{j\in J}(x-\Ga_j)^{-v_j})\ge d$, i.e., $\sum_{i\in I}v_i\ge d$ or $\sum_{j\in J}-v_j\ge d$. This gives $\sum_{i=1}^n|v_i|\ge d$. Therefore, $||\bv||=\sqrt{\sum_{i=1}^nv_i^2}\ge \sqrt{\sum_{i=1}^n|v_i|}\ge \sqrt{d}$ (note that each $v_i$ is an integer).

If there exists a lattice point with $d$ nonzero entries which are either all $1$ or $-1$, then it is clear that $\lambda_1(\mL_{\ba,c(x)}) = \sqrt{d}$. Conversely, assume $\lambda_1(\mL_{\ba,c(x)}) = \sqrt{d}$. Then there exists a nonzero lattice point $\bv = (v_1,v_2,\ldots,v_n)$  in $\mL_{\ba,c(x)}$ such that $||\bv||=\sqrt{\sum_{i=1}^nv_i^2}=\sqrt{d}$. Since $\deg(\prod_{i\in I}(x-\Ga_i)^{v_i}-\prod_{j\in J}(x-\Ga_j)^{-v_j})\ge d$, we must have that either $I=\emptyset$ \& $\sum_{j\in J}-v_j= d$ or $J=\emptyset$ \& $\sum_{i\in I} v_i= d$. This forces that either  $v_i=1$ for all $i\in I$ or $v_j=-1$ for all $j\in J$.
\end{itemize}
\end{proof}

According to Lemma \ref{lattice_properties_lem} (iii), we see that $\mL_{\ba,c(x)}$ has minimum norm $\lambda_1(\mL_{\ba,c(x)}) = \sqrt{d}$ when there exist $i_1, \ldots, i_d \in [n]$ such that $\prod_{j=1}^d(x-\alpha_{i_j}) = 1 +c(x)$. It follows that there are at most ${n \choose d}$ different $c(x) \in \fR$ of degree $d$ out of a total of $q^d$ such polynomials such that $\mL_{\ba,c(x)}$ satisfies $\lambda_1(\mL_{\ba,c(x)}) = \sqrt{d}$. In other words, given a polynomial $c(x)$ of degree $d$ and an ordered set $(\alpha_1, \ldots, \alpha_n)$, the probability that the lattice $\mL_{\ba, c(x)}$ has minimum norm $\sqrt{d}$ is less than $1/d!$ and we can expect the minimum norm of the lattice $\mL_{\ba,c(x)}$ to be bigger (if $d$ is small). In particular, we will use the Gaussian heuristic to estimate the minimum norm of the lattices.  Assume that the map $\phi_\ba$ is surjective. By Lemma \ref{lattice_properties_lem} (ii) and Lemma \ref{unitgroup_lem}, the determinant of $\mL_{\ba,c(x)}$ is approximately $q^d$. The Gaussian heuristic suggests that a random lattice of dimension $n$ and determinant $q^d$ has minimum norm approximately $\sqrt{n/2\pi e}q^{d/n}$.




Next, we describe how to construct the ordered set $\ba = (\alpha_1, \ldots, \alpha_n)$ for which $\mL_{\ba,c(x)}$ admits a nice basis for a class of $c(x)$. In the following, we assume that $c(x)$ is of the form $c(x) = c_1(x)\ldots c_t(x)$, where $c_i(x)$'s are pairwise coprime irreducible polynomials over $\F_q$, each having degree $d_0$. Hence, $\fQ_{c(x)} \cong \oplus_{i=1}^t\F_{q^{d_0}}$. Let $\beta$ denote a generator of $\F_{q^{d_0}}$.

Let $\alpha_{n-t+1}, \ldots, \alpha_n$ be $t$ distinct elements in $\F_q$. For $i=1, \ldots, t$ and $j=1, \ldots, t$, let $\gamma_{ij} = x-\alpha_{n-t+i} \mod c_j(x)$.
Let $m_{ij} = \log_\beta{\gamma_{ij}}$ and $M = (m_{ij})_{i=1,\ldots,t, j=1, \ldots, t}$.

Suppose that $M$ is invertible over the ring $\Z_{q^{d_0}-1}$. For each $\alpha \in \F_q$ with $\alpha \ne \alpha_{n-t+1}, \ldots, \alpha_n$, let $\vec{y} = (y_1, \ldots, y_t)$, where $y_j = \log_{\beta}((x-\alpha) \mod c_j(x))$, $j=1, \ldots, t$. Let $\vec{g_{\alpha}} = \vec{y}M^{-1} \mod q^{d_0}-1$. Write $\vec{g_{\alpha}} = (g_{\alpha,n-t+1}, \ldots, g_{\alpha,n})$. Note that for each $j=1, \ldots, t$,
$$y_j = \sum_{i=1}^tg_{\alpha,n-t+i}m_{ij} \mod q^{d_0}-1.$$

For $j=1, \ldots, t$, we have

\begin{align*}
\prod_{i=1}^t(x-\alpha_{n-t+i})^{g_{\alpha,n-t+i}} & \mod c_j(x) \equiv \prod_{i=1}^t(\beta^{m_{ij}})^{g_{\alpha,n-t+i}} \mod c_j(x) \\
\equiv \beta^{\sum_{i=1}^tg_{\alpha,n-t+i}m_{ij}} & \mod c_j(x) \equiv \beta^{y_j} \mod c_j(x) \equiv x-\alpha \mod c_j(x).
\end{align*}

Since it holds for any $c_j(x)$, it follows that $x-\alpha \equiv \prod_{i=1}^t(x-\alpha_{n-t+i})^{g_{\alpha,n-t+i}} \mod c(x).$
Consequently, the point $(0,\ldots,1,0,\ldots,-g_{\alpha,n-t+1},\ldots,-g_{\alpha,n})$, where $1$ is in the entry indexed by $\alpha$ is a point in $\mL_{\ba,c(x)}$ for any $\alpha \in (\alpha_1, \ldots, \alpha_{n-t})$.

\begin{proposition}
\label{basis_prop}
Let $\alpha_{n-t+1}, \ldots, \alpha_n$ be $t$ distinct elements in $\F_q$ with the matrix $M$ as above. Suppose that $M$ is invertible over $\Z_{q^{d_0}-1}$. Pick $\alpha_1, \ldots, \alpha_{n-t}$ randomly from $\F_q$ such that $\ba = (\alpha_1, \ldots, \alpha_n)$ contains $n$ distinct elements. Define $G$ as the $(n-t) \times t$ matrix with rows given by $\vec{g_{\alpha_i}}$, for $i=1, \ldots, n-t$. A basis of the lattice $\mL_{\ba,c(x)}$ is given by:
$$B_{\ba,c(x)} = \left(\begin{matrix}
I_{n-t}&-G\\
0_{t\times (n-t)} & (q^{d_0}-1)I_t
\end{matrix}\right),$$
where $I_r$ denotes the identity matrix of rank $r$.
\end{proposition}

\begin{proof}
According to the preceding arguments, the first $n-t$ rows of $B_{\ba,c(x)}$ are points in $\mL_{\ba,c(x)}$. Since $(x-\alpha_i)^{q^{d_0}-1} \equiv 1 \mod c(x)$ for $i=n-t+1, \ldots, n$, the last $t$ rows of $B_{\ba,c(x)}$ are also in $\mL_{\ba,c(x)}$. Clearly, the rows of the matrix are linearly independent. It remains to show that the rows span $\mL_{\ba,c(x)}$. Let $\vec{u}=(u_1,\ldots,u_n)$ be a point in $\mL_{\ba,c(x)}$ so that $\prod_{i=1}^n(x-\alpha_i)^{u_i} \equiv 1 \mod c(x)$. Consider the point $\vec{v} = \vec{u} - \sum_{i=1}^{n-t}u_iB_i$, where $B_i$ denotes the $i$-th row of $B_{\ba,c(x)}$. Hence, $\vec{v} \in \mL_{\ba,c(x)}$ and we can write $\vec{v} = (0,\ldots,0,v_{n-t+1}, \ldots, v_n)$. It is sufficient to show that $v_{n-t+i} \equiv 0 \mod q^{d_0}-1$ for $i=1, \ldots, t$. In other words, $\prod_{i=1}^t(x-\alpha_{n-t+i})^{v_{n-t+i}} \equiv 1 \mod c(x)$, equivalently, $\prod_{i=1}^t(x-\alpha_{n-t+i})^{v_{n-t+i}} \equiv 1 \mod c_j(x)$ for $j=1, \ldots, t$.
Now, $$\prod_{i=1}^t(x-\alpha_{n-t+i})^{v_{n-t+i}} \equiv \prod_{i=1}^t(\beta^{m_{ij}})^{v_{n-t+i}} \equiv \beta^{\sum_{i=1}^tv_{n-t+i}m_{ij}} \equiv 1 \mod c_j(x)$$
which gives $(v_{n-t+1}, \ldots, v_n)M = 0 \mod q^{d_0}-1$. Since $M$ is invertible, we conclude that $v_{n-t+i} \equiv 0 \mod q^{d_0}-1$ for $i=1, \ldots, t$.
\end{proof}

\begin{remark}
Note that the lattices $\mL_{\ba, c(x)}$ for different pairs of $\ba$ and $c(x)$ are not all distinct. For instance, let $\ba = (\alpha_1, \ldots, \alpha_n)$ and let $\gamma \ne 0 \in \F_q$. Let $\ba' = (\alpha_1+\gamma, \ldots, \alpha_n+\gamma)$ and $c'(x) = c(x-\gamma)$. Then, it is easy to check that $\mL_{\ba, c(x)} = \mL_{\ba', c'(x)}$.
\end{remark}

Next, we analyze the complexity of constructing $B_{\ba,c(x)}$. First, one needs to compute about $tn$ discrete logs in the field $\F_{q^{d_0}}$. The discrete logarithm problem over finite fields is one of the fundamental hard problems widely used in cryptography. Extensive studies have been done in this area and various methods have been proposed to solve the discrete logarithm problem over finite fields. In particular, it is adequate for us to employ Pollard's rho method to compute discrete logarithm with time complexity $O(\sqrt{q^{d_0}})$.
Please refer to the survey paper \cite{JouxP16} for the state-of-the-art results on the discrete logarithm problem.
For our construction, we have $r = q^{d_0}$. For $q=O(n)$ and $d_0=O(1)$, it follows that solving the discrete log is efficient.

Second, one needs to pick $\alpha_{n-t+1}, \ldots, \alpha_n$ so that the matrix $M$ is invertible. Now, each entry $m_{ij}$ is the discrete log of $x-\alpha_{n-t+i} \mod c_j(x)$. Since $\alpha_{n-t+i}$ is random, we may assume that the matrix $M$ is a random matrix in the ring $\Z_{q^{d_0}-1}$.


\begin{lemma}\cite[Theorem 2]{Ngu99}
Let $s = q^{d_0}-1$ be a positive integer. Let $p_1, \ldots, p_m$ be the distinct prime divisors of $s$. The probability that a random $t \times t$ matrix in $\Z_s$ is invertible is
$$P_s = \prod_{i=1}^{m} \prod_{j=1}^{t} (1-p_i^{-j}).$$
\end{lemma}

It can be seen from the formula that the probability of a random $t\times t$ matrix being invertible converges to a constant for large dimension $t$. In Figure \ref{fig:PrOfInvertMatrix}, we give the probability to obtain a random nonsingular matrix with modulus $s \in \{2, \ldots, 300\}$ and a fixed dimension $t=200$. From the results, it can be seen that $P_s$ is non-negligible for this range of modulus.

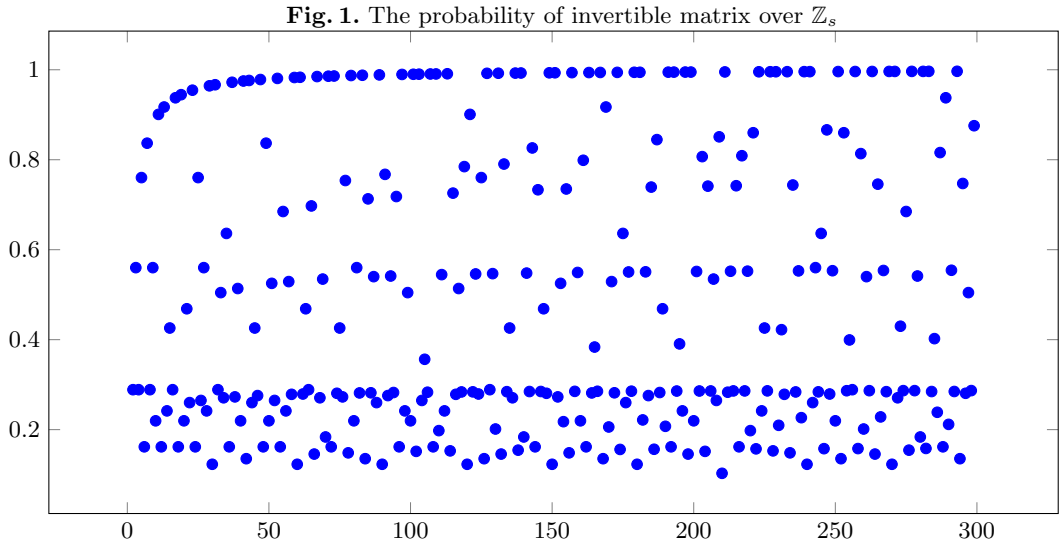
\begin{figure}
\caption{The probability of invertible matrix over $\Z_s$}
\begin{tikzpicture}[auto, left]
\begin{axis}[
width=15cm,
height=8cm,
]

\addplot[only marks, blue] file {PlotData.data};
\end{axis}
\end{tikzpicture}
\label{fig:PrOfInvertMatrix}
\end{figure}


\section{Construction of Our Trapdoor Functions}\label{sec:trapdoor function}

In this section, we describe new trapdoor functions where inverting the function amounts to solving the CVP for the associated lattices. Unlike the GGH construction, we do not generate two different bases of a lattice. Instead, we require only one basis of our polynomial lattice as the trapdoor involves information to construct the polynomial lattice. Recall that a trapdoor function encompasses four different sub-algorithms, namely, {\em generate}, {\em sample}, {\em evaluate} and {\em invert}. We will now present each of these in detail.

{\bf Generate:}
Set the public parameters $q,n,d$ according to the desired security level (see the next section for details). Let $d = d_0t$. Choose $t$ irreducible polynomials $c_i(x)$ of degree $d_0$ and let $c(x) = c_1(x)\ldots c_t(x)$. Choose an ordered  set $\ba = (\alpha_1, \ldots, \alpha_n)$ such that $c(\alpha_i) \ne 0$ for $i=1, 2, \ldots, n$ and the elements $\alpha_{n-t+1}, \ldots, \alpha_n$ satisfy the conditions in Proposition \ref{basis_prop}. Construct the basis $B_{\ba,c(x)}$ of the lattice $\mL_{\ba,c(x)}$ as described in Proposition \ref{basis_prop}. Write $B_{\ba,c(x)} = \left(\begin{matrix}I_{n-t}&-G\\
0_{t\times (n-t)} & (q^{d_0}-1)I_t\end{matrix}\right).$

Let $H = B'_{\ba,c(x)} = \left(\begin{matrix}I_{n-t}, & -G\end{matrix}\right)$.
The trapdoor for our function includes the polynomial $c(x)$ and the ordered set $\ba = (\alpha_1, \ldots, \alpha_n)$.

{\bf Sample:}
Randomly sample $\vec{m} \in \mathbb{Z}_{q^{d_0}-1}^{n-t}$ and the error $\vec{e} = (e_1, \ldots, e_n)  \in \{0,1\}^n$ satisfying:
$\sum_{i=1}^ne_i = d-1$.

{\bf Evaluate:}
For each input $\vec{m} \in \Z_{q^{d_0}-1}^{n-t}$, the function $f$ is evaluated on $\vec{m}$ as
$$\vec{c}=f(\vec{m}, \vec{e}) = \vec{m}H + \vec{e} \mod q^{d_0}-1.$$

{\bf Invert:}
Suppose that we are given a valid output $\vec{c} = (c_1, \ldots, c_n)$ of the function $f$. The inversion process is as follows.
\begin{itemize}[itemindent=3em]

\item[{\rm Step 1:}] Compute
$$r(x) = \prod_{i=1}^n(x-\alpha_i)^{c_i} \mod c(x).$$

\item[{\rm Step 2:}] Factorize $r(x)$ as $r(x) = \prod_{i=1}^n(x-\alpha_i)^{u_i}$. Let $\vec{u} = (u_1, u_2, \ldots, u_n).$

\item[{\rm Step 3:}] Compute $\vec{v'} = \vec{c}-\vec{u}$. Write $\vec{v'} = (v_1', \ldots, v_n')$.

\item[{\rm Step 4:}] Let $\vec{m'}=(v_1', \ldots, v_{n-t}')$.
\end{itemize}

Without knowledge of the trapdoor, observe that inverting the function will require us to find the error $\vec{e}$ or equivalently, a point in $\mL_{\ba, c(x)}$ that is close to $\vec{c}$. Concretely, one will use the basis formed by the rows of the matrix $ \left(
\begin{array}{cc}
&\mbox{\smash{\textit{H}}}\\
0&(q^{d_0}-1)I_t
\end{array}
\right).$
Thus, one needs to be able to solve CVP with respect to this basis. We will discuss more about this in the next section.

The following theorem shows that the inversion process indeed recovers $\vec{m}$.

\begin{theorem}
\label{invert_thm}
Let $\vec{m}$ be a random element in $\Z_{q^{d_0}-1}^n$ and let $\vec{c}$ be the output produced by the {\bf Evaluate} algorithm. Let $\vec{m'}$ be the output of the {\bf Invert} algorithm. Then $\vec{m'}=\vec{m}$.
\end{theorem}

\begin{proof}
First, we have $\vec{c} = \vec{m}B_{\ba,c(x)}'+\vec{e}$. We claim that $\vec{v} = \vec{v'}$, where $\vec{v} = \vec{m}B_{\ba,c(x)}'$.
To see this, note that

\begin{align*}
\prod_{i=1}^n(x-\alpha_i)^{c_i} = & \prod_{i=1}^n(x-\alpha_i)^{v_i+e_i} = \prod_{i=1}^n(x-\alpha_i)^{v_i}\prod_{i=1}^n(x-\alpha_i)^{e_i} \\
\equiv & 1 \cdot \prod_{i=1}^n(x-\alpha_i)^{e_i} \mod c(x) \equiv \prod_{i=1}^n(x-\alpha_i)^{e_i} \mod c(x)
\end{align*}

Since $\sum_{i=1}^ne_i =d-1 < d$, we must have $r(x) = \prod_{i=1}^n(x-\alpha_i)^{e_i}$, so $u_i = e_i$.
Therefore, $\vec{v'}=\vec{c}-\vec{u}=\vec{m}B_{\ba, c(x)}'+\vec{e}-\vec{u}=\vec{v}$ and the claim is proved.

Note that we have $\vec{v} = \vec{m}B_{\ba, c(x)}' = (\vec{m}, -\vec{m}G)$. Therefore, we have $(v_1, \ldots, v_{n-t}) = \vec{m} \mod q^{d_0}-1$.
\end{proof}

\begin{remark}

\begin{itemize}
\item In general, we like to have as many nonzero entries of the error as possible. Hence, we choose $e_i$ to take small values. In particular, we typically let $e_i = 1$.

\item Instead of letting all the error entries be positive, we can equivalently let them be all negative. In this case, in the inversion process, one needs to check if $r(x)$ or $1/r(x) \mod c(x)$ can be factorized. In the former case, we have the usual case where $e_i \ge 0$. In the latter case, it is easy to verify that we have $e_i \le 0$, that is all the nonzero entries of the error are $-1$.

\item For the inversion process, one can simply check if $r(\alpha_i) = 0$ to check if $u_i=0$ or $1$.
\end{itemize}
\end{remark}



\begin{remark}
\begin{itemize}

\item Here, only the right part $-G$ of the matrix $H = \left(\begin{matrix}I_{n-t}, & -G\end{matrix}\right)$, which is used to evaluate the function, is undetermined. It is an $(n-t)\times t$ matrix over $\Z_{q^{d_0}-1}$ and thus, has size $(n-t)td_0\log_2q$ bits.

\item Unlike the GGH scheme, inversion does not require solving the CVP. Instead, inversion is carried out using properties of polynomials and remainders.

\item In the above scheme, the first $n-t$ positions of $\vec{c}$ may contain some information about $\vec{m}$. This is because we have only introduced error to $d-1$ positions, and thus, at least $n-t-d+1$ positions will be in the clear. In Section \ref{sec:practical_scheme}, we present a practical encoding scheme to mask the original message $m$.

\item Apart from the above scheme, other modifications are possible. Randomly pick an $(n-t)\times (n-t)$ unimodular matrix $T$ with small entries and an $n \times n$ permutation matrix $P$. Construct $H = TB_{\ba, c(x)}'P \mod {q^{d_0}-1}$. The left part of the new matrix $H$ will hide all the information of message $m$. In situations where the inputs are completely random, the roles of $T$ and $P$ will not be so critical.


\end{itemize}
\end{remark}

\section{Security Analysis of Our Trapdoor Functions}\label{sec:security}

In deciding the parameters for our scheme, we will like to achieve the following:
\begin{itemize}
\item The public key size should be reasonably small;
\item Key generation, encryption and decryption should be efficient;
\item The scheme is resistant against all existing attacks.
\end{itemize}

We now discuss some possible attacks on our scheme to help us decide the appropriate parameters. First, suppose that $d \ge n/2+1$. Let $\vec{c}$ be a valid output with error $\vec{e}$. Then, $\vec{e}$ has $d-1$ entries $=1$. Consider $\vec{c'}=\vec{c}-(1, \ldots, 1)$. It is $\vec{c'}=\vec{m}H+\vec{e} - (1,\ldots,1) = \vec{m}H+\vec{e'}$, where $\vec{e'}$ has $< n/2$ entries $= -1$. Hence, one may decrypt using $\vec{c'}$ instead. It follows that we may assume that $d \le n/2$. Therefore, we have $t \le d \le n/2$ and $n \le q$.

\subsection*{Error search}

At first glance, it appears that one needs to search through all ${n \choose d-1}$ entries to find the error. However, one can in fact reduce the search in the following way. It is obvious that the first $n-t$ columns of $H$ are linearly independent. Let $I = \{1, \ldots, {n-t}\}$. Search through all possible error positions in $I$. Recall that there are at most $d-1$ such positions. Assuming that the error bits are uniformly distributed, the number of nonzero error bits in these positions is roughly $l=\frac{(n-t)(d-1)}{t}$. Consequently, the number of tries is around ${n-t \choose l}$.

Let $\vec{c_I}$ denote the vector formed by the entries in $\vec{c}$ indexed by $I$. For each guess $\vec{e_I}$, define $\vec{x} = \vec{c_I}-\vec{e_I} \mod q^{d_0}-1$. To verify if our guess is correct, check if $\vec{c}-\vec{x}H$ is of the correct error form, that is, contains exactly $d-1$ $1$'s and all other entries are $0$.

Consequently, the complexity of this attack is $O({n-t \choose l}(n-t)t)$.

\subsection*{Search for the trapdoor}

One obvious way to attack the function is to find the trapdoor information. We will need to search for $c(x)$ and $\ba = (\alpha_1, \ldots, \alpha_n)$. One way to do this is as follows:

\begin{itemize}
\item Exhaustively search for the polynomial $c(x)$. There are $O(q^d)$ different $c(x)$ of degree $d$ of the form $c(x) = c_1(x)\ldots c_t(x)$.

\item For each $c(x)$, guess the ordered set $(\alpha_{n-t+1}, \ldots, \alpha_n)$. For each such set, determine if there exist $\alpha_1, \ldots, \alpha_{n-t}$ that satisfy the matrix $B_{\ba, c(x)}' = \left(\begin{matrix}I_{n-t}, & -G\end{matrix}\right)$. Let $-G=(b_{i, n-t+j})_{i=1,\ldots,t, j=1,\ldots, t}$. Specifically, from the definition of $\mL_{\ba, c(x)}$ and $B_{\ba, c(x)}'$, we can construct $\alpha_i$ by checking for $\alpha_i$ such that $(x-\alpha_i) \times \prod_{j=1}^t{(x-\alpha_{n-t+j})^{b_{i, n-t+j}}} \equiv 1 \pmod {c(x)}$ for each $i \in \{1, \ldots, n-t\}$. Our guesses of $c(x)$ and $\alpha_i$'s are correct if we can reconstruct $(\alpha_1, \ldots, \alpha_{n-t})$ by the preceding procedure. There are $\Perms{n}{n-t} \approx n^{n-t}$ possible ordered sets $(\alpha_{n-t+1}, \ldots, \alpha_n)$.
\end{itemize}

The overall complexity of this attack is $O(q^d n^{n-t} (n-t) d^3 (\log q)^3)$ if we assume that the complexity of polynomial multiplication in the ring $\F_q[x]/c(x)$ is $O(d^2 (\log q)^2)$.



\subsection*{Inverting the function via solving CVP}

Next, we discuss the effectiveness of inverting the trapdoor function by solving CVP. As mentioned in Section \ref{sec:background}, solving the CVP for a random lattice and a random target vector is hard. In our situation, the error vector is of a special form, namely, it contains exactly $d-1$ $1$'s or $-1$'s. We first investigate how well Babai's nearest plane algorithm works to recover the error.
In Table \ref{table:babai_experiment_result}, we give some experimental results when Babai's nearest plane algorithm is used to invert a random instance of our trapdoor function. In our experiments, we let $d=t$, which means that $c(x)$ is a product of linear polynomials. In addition, we let $q$ be the next prime number larger  than $n+d$. Before running Babai's nearest plane algorithm, we converted the basis to a BKZ-$\beta$ reduced basis, where $\beta$ is the block size involved \footnote{Every BKZ algorithm related experiment conducted in this paper is run under the SageMath software with parameter \texttt{proof=False}, which calls the fplll library.}. For each $(n, d)$ pair, we repeated the experiments $30$ times. The status \texttt{T} in Table \ref{table:babai_experiment_result} means that there is at least one successful inversion among the repeated experiments with the same set of parameters ($n$, $q$, $d$ and block size), while \texttt{F} means that no successful inversion was achieved. In  Table \ref{table:babai_experiment_result}, for each pair of $n$ and block size, we provide the largest value of $d$ that results in the status \texttt{F} and the smallest $d$ (which is necessarily the next value) that results in the status \texttt{T}.


From the experimental results, we see that, for any fixed $n$ and BKZ block size, the attack by Babai's nearest plane algorithm is more effective for larger $d$. On the other hand, for any fixed $d$ and BKZ block size, this attack becomes ineffective as $n$ increases. Finally, for fixed $n$ and $d$, one may increase the BKZ block size to attempt to invert the function. However, it appears that the impact is minimal when the BKZ block size is increased beyond a certain bound for each fixed $(n,d)$ pair. In particular, for $n \ge 200$, our results suggest that Babai's algorithm will not be effective for practical block sizes when $25\leq d\leq 40$.

\begin{table}[ht!]
\centering
\caption{Experimental results on Babai's algorithm to invert the trapdoor function}
\begin{tabular}{|c|cc|cc|cc|cc|cc|cc|cc|cc|cc|}
\hline
status     &F &T &F &T &F  &T  &F  &T  &F  &T  &F  &T  &F  &T  &F  &T  &F  &T   \\ \hline
$n$        &80&80&80&80&100&100&100&100&100&100&100&100&100&100&100&100&100&100 \\ \hline
$d$        &26&27&24&25&39 &40 &33 &34 &30 &31 &30 &31 &28 &29 &28 &29 &28 &29  \\ \hline
block size &20&20&30&30&20 &20 &30 &30 &40 &40 &45 &45 &50 &50 &55 &55 &60 &60  \\ \hline
\end{tabular}
\label{table:babai_experiment_result}
\end{table}



\subsection*{Embedding attack to find error} \label{sec:embedding_attack}

In \cite{Ngu99}, by exploiting the structure of the errors of the GGH scheme, the embedding attack was employed to break the scheme with $n$ up to $350$. It was suggested that the embedding attack is effective whenever the gap between the minimum norm of $L$ and the error norm is too big. Extensive experiments were carried out in \cite{GN08} to analyze the effectiveness of the embedding attack with respect to this gap. It was proposed that in order for the attack to be effective via BKZ algorithms,  one should choose a block size with corresponding $\delta$ satisfying $\lambda_1(L)/{\rm error norm} > \epsilon \delta^n$, where $\delta$ is the root Hermite factor and $\epsilon$ is some small constant. The value of $\epsilon$ is not known for a random lattice. In \cite{GN08}, experiments were carried out on semi-orthogonal lattices as well as knapsack lattices with both the LLL and BKZ-$20$ algorithms. We performed extensive experiments to estimate an appropriate value of $\epsilon$ for our lattices.
In Appendix \ref{appendix:more_embedding_attack_data}, we provide our experimental results that guide us to choose a suitable $\epsilon$.

Once $\epsilon$ is fixed, one can resist the embedding attack by choosing the parameters so that the $\delta$ required to launch a successful attack will be infeasible to achieve.  In our situation, with $\lambda_1(\mL_{\ba,c(x)})$ estimated by the Gaussian heuristic, we have

\begin{equation*}
\begin{aligned}
\sqrt{n/(2\pi e(d-1))}(q^{d_0}-1)^{t/n} \le \epsilon\delta^n.
\end{aligned}
\label{eq:embedding_attack_formula}
\end{equation*}


As in the attack via Babai's algorithm, we carried out some experiments to investigate how well our lattices and errors can withstand the embedding attack. First, we performed experiments to compute the root Hermite factor of our bases for different BKZ block sizes. Once again, we let $t=d$ and $q$ be the next prime larger than $n+d$. We set $n \in \{149, 150, 151\}$ and $d \in \{70, \ldots, 80\}$. For each $(n, d)$ pair, we also repeated the experiments $30$ times. Then we picked the smallest value of $\delta$ as the root Hermite factor indicated in Table \ref{table:expertment_block_size_delta}.

\begin{table}[!ht]
\centering
\caption{Experiments on block size and root Hermite factor $\delta$}
\begin{tabular}{|c|c|c|c|c|c|c|c|c|c|}
\hline
block size & 20      & 30      & 40      & 50      & 55      & 60      & 65      & 70 \\ \hline
$\delta$   & 1.01168 & 1.01135 & 1.01119 & 1.01098 & 1.01007 & 1.00987 & 1.00934 & 1.00902 \\ \hline
\end{tabular}
\label{table:expertment_block_size_delta}
\end{table}

Next, we present our experimental results on the embedding attack. In these experiments, we fix $n = 150$ and let the BKZ block size $\beta$ vary. Our choices of $d$ and $q$ are identical to those in the previous experiments. For each instance, we also repeated the experiments $30$ times. One successful embedding attack represents the status \texttt{T} in the first column. Otherwise, we label the status as \texttt{F}. In these experiments, we find the largest value of  $d$ that can resist the embedding attack as indicated in Table \ref{table:embedding_attack_result}. Our results show that as the block size increases, the maximum value of $d$ that can resist the embedding attack decreases. However, increasing the block size will involve a much longer basis reduction time. Using the  experimental value of  $\delta$ obtained in Table \ref{table:expertment_block_size_delta}, we compute the corresponding $\epsilon$ from the embedding attack formula in the last column.


\begin{table}[ht!]
\centering
\caption{Embedding attack experimental results}
\begin{tabular}{|c|p{0.8cm}<{\centering}|p{0.6cm}<{\centering}|p{0.8cm}<{\centering}|c|c|c|}
\hline
status & $n$&$d$ &$q$& block size &$\delta$&experimental $\epsilon$ \\ \hline
F      & 150& 66 & 223 & 20         &1.01168& 0.69386 \\
T      & 150& 67 & 223 & 20         &1.01168& 0.71383 \\ \hline
F      & 150& 58 & 211 & 30         &1.01135& 0.57091 \\
T      & 150& 59 & 211 & 30         &1.01135& 0.58651 \\ \hline
F      & 150& 55 & 211 & 40         &1.01119& 0.53973 \\
T      & 150& 56 & 211 & 40         &1.01119& 0.55421 \\ \hline
F      & 150& 50 & 211 & 50         &1.01098& 0.48910 \\
T      & 150& 51 & 211 & 50         &1.01098& 0.50176 \\ \hline
F      & 150& 44 & 197 & 55         &1.01007& 0.47288 \\
T      & 150& 45 & 197 & 55         &1.01007& 0.48421 \\ \hline
F      & 150& 43 & 197 & 60         &1.00987& 0.47586 \\
T      & 150& 44 & 197 & 60         &1.00987& 0.48713 \\ \hline
F      & 150& 42 & 197 & 65         &1.00934& 0.50307 \\
T      & 150& 43 & 197 & 65         &1.00934& 0.51484 \\ \hline
F      & 150& 39 & 191 & 70         &1.00902& 0.48913 \\
T      & 150& 40 & 193 & 70         &1.00902& 0.50141 \\ \hline
\end{tabular}
\label{table:embedding_attack_result}
\end{table}

In the asymptotic case, in the survey paper of \cite{APS15}, a general relationship between $\delta$ and  block size $\beta$ is given as $\delta \approx \beta^{1/2\beta}$. In addition, the time complexity to run the BKZ algorithm is estimated by the following result.

\begin{proposition}
\label{BKZ_prop}
The log of the time complexity for running BKZ to achieve a root Hermite factor $\delta$ is:
\begin{itemize}
    \item $\Omega(\frac{\log^{2}(\log\delta)}{\log^{2}\delta})$ if calling the SVP oracle costs $2^{\mathcal{O}(\beta^{2})}$,
    \item $\Omega(\frac{-\log(\frac{-\log\log\delta}{\log\delta}) \log\log\delta}{\log\delta})$ if calling the SVP oracle costs $\beta^{\mathcal{O}(\beta)}$,
    \item $\Omega(\frac{-\log\log\delta}{\log\delta})$ if calling the SVP oracle costs $2^{\mathcal{O}(\beta)}$.
\end{itemize}
\end{proposition}


When $n$ is large, we will like to have $d \approx n/2$. Let $q$ to be slightly bigger than $n+d$, say $q \approx 2n$. In order for $\delta$ to satisfy $\sqrt{n/2\pi e(d-1)}q^{d/n} > \epsilon\delta^n$, we have
$$\delta < (2n/\pi e\epsilon^2)^{1/2n}.$$

This gives $\delta$ close to $1$ when $n$ is large and consequently, we need a block size close to $n$.

For smaller values of $n$, \cite[Table 2]{CN11} gives some estimates for $\delta$ corresponding to $\beta \le 250$ achieved using their BKZ2.0 algorithm. In addition, they provided time estimates to run the algorithm by measuring the enumeration cost \cite[Table 3, 4]{CN11} to run the SVP sub-routine and the number of BKZ rounds required. 
Note that the total cost of BKZ is estimated to be $(n-1)*{\rm number of rounds}*{\rm enumeration cost}$. In the following section, we will use this method to give the estimated cost for some parameters.



\subsection*{Enhanced embedding attack}

The embedding attack can be enhanced by combining with partial search of the error bits. Specifically, if $k$ nonzero error bits are guessed correctly, the remaining error norm will be reduced to $\sqrt{d-k}$, thereby making the gap from $\lambda_1(\mL_{\ba,c(x)})$ bigger. Concretely, we have the new gap as $\sqrt{n/2\pi e (d-1-k))}(q^{d_0}-1)^t$. This in turn may reduce the BKZ block size needed to launch the embedding attack. Hence, we need to ensure that ${n \choose k}$ times of each single BKZ execution will be infeasible to carry out.

We will illustrate the enhanced embedding attack by a concrete example by referring to our parameter choices in the next section. Suppose that the adversary has correctly guessed $k=d/2$ errors for the first row of the practical parameters we present in Table \ref{table:estimated_cost}. Then $n=230, t=d=29, k=15, d_0=1, q=263$. Now, the adversary needs to recover the remaining error bits via the embedding attack. This requires the adversary to run a BKZ algorithm with $\delta=1.0084$ for the block size $\beta \approx 133$, which corresponds to the cost $2^{170.24}$ by the method in Section \ref{sec:embedding_attack}. In fact, the cost of the adversary to correctly guessed the $k$ errors is ${{n-d}\choose{k}} \approx 2^{73.7}$. Hence, the enhanced embedding attack is infeasible to carry out for our practical parameters. Using a similar argument, One can check that this approach does not work for the remaining proposed parameters as well.

\subsection*{Other attacks to find the trapdoor information}

Note that with knowledge of the public information $G$, one can easily construct the matrix $B_{\ba,c(x)}$. The question is whether this matrix will leak information about the polynomial $c(x)$ as well as $\ba$. Here, we discuss a possible attack when $c(x)$ is irreducible over $\F_q$, that is, $t=1$.

In this case, the matrix $B_{\ba,c(x)}$ takes a very simple form, namely,

$$B_{\ba,c(x)} = \left(\begin{matrix}I_{n-1}&-G\\
0_{1\times (n-1)}&q^d-1\end{matrix}\right),$$
where $G$ is a $(n-1) \times 1$ column. Write $G$ as $G = \left(\begin{matrix}g_1\\g_2\\\vdots\\g_{n-1}\end{matrix}\right),$ where each $g_i \in \Z_{q^d-1}$.
Note that $g_i$ satisfies $x-\alpha_i \equiv (x-\alpha_n)^{g_i} \mod c(x)$. Without any loss of generality, we may assume that $\alpha_n=0$ (by substituting $x$ by $x-\alpha_n$ in the whole system). It follows that $c(x)$ is a common factor of the polynomials $x-\alpha_i - x^{g_i}$, $i=1, 2, \ldots, n-1$. Since $c(x)$ is irreducible over $\F_q$ of degree $d$, it is a factor of $x^{q^d}-x$.

We can now perform the following steps to recover $c(x)$ and the $\alpha_i$'s.

\begin{itemize}
\item Randomly guess $\alpha_1$. For each $\alpha_1$, compute the gcd $h(x) = \gcd(x-\alpha_1-x^{g_1}, x^{q^d}-x)$. Find all pairs $\alpha_1,d(x)$ such that $d(x)$ is irreducible of degree $d$ and divides the polynomial $h(x)$.

\item For each pair $\alpha_1, d(x)$ found above, test for $\alpha_2$ such that $d(x)$ is also a factor of $x-\alpha_2-x^{g_2}$.

\item Continue the process until one $d(x)$ is left. Let $c(x)$ to be this $d(x)$.

\item Find the remaining $\alpha_i$'s by direct computation of $\alpha_i = x-x^{g_i} \mod d(x)$.
\end{itemize}

We remark that with high probability, the set of possible $d(x)$ after the first step will be very small. It follows that the main complexity of the above attack comes from performing the gcd computations to find gcd of polynomials of the form $x-\alpha -x^g$ and $x^{q^d}-x$. In general, such a gcd computation has complexity polynomial in $g$. Furthermore, with high probability, $g$ is of the order of $q^d$. Consequently, in general, the above attack has complexity polynomial in $q^d$. However, the above attack works if $g$ is small or is of a special form that makes the gcd computation easy.

The above attack easily generalizes to the case when $t > 1$ but the complexity increases as well. Specifically, we will need to guess $t$ different values of $\alpha$ in the first step. This has complexity $n!/(n-t)! \approx n^t$. In view of these considerations, we will choose $t$ to be as big as possible, say $c(x)$ is a product of linear or quadratic polynomials.

\begin{remark}
Apart from Babai's algorithm, one may use enumeration methods with pruning to solve CVP. Our preliminary experiments showed that pruning techniques do not have a great advantage over Babai's algorithm for our lattices, particularly when $n \ge 200$.
\end{remark}

\section{A Practical Encryption Scheme}

\subsection{Description} \label{sec:practical_scheme}

Similar to the GGH encryption scheme and the McEliece encryption scheme, in order to transit from the one-way trapdoor function to an encryption scheme, one needs a method to encode the message before parsing to the trapdoor function. In particular, the chosen encoding scheme should ensure that the encryption scheme is semantically secure. Different proposals were presented in \cite{GGH97,KH01} to achieve semantic security for the GGH scheme and the McEliece scheme. We first show why the scheme employed in \cite{GGH97} will not work for our construction.

Recall that for the GGH scheme, it was suggested to encode the plaintext bits as the least significant bits of the input message to the trapdoor function and the other bits are allowed to be picked randomly. We now show how this will make our scheme vulnerable to the related message attack.

To this end, let $\vec{p}$ be an $(n-t)$-bit plaintext to be encrypted. Suppose that $\vec{p}$ is encrypted twice, that is, encoded into $\vec{m_1}$ and $\vec{m_2}$ with $\vec{p}$ occupying the least significant bits of $\vec{m_1}$ and $\vec{m_2}$. This gives $\vec{c_1}=\vec{m_1}H+\vec{e_1} \mod q^{d_0}-1$ and $\vec{c_2}=\vec{m_2}H+\vec{e_2} \mod q^{d_0}-1$. Summing up, this yields $(\vec{m_1}+\vec{m_2})H+\vec{e_1}+\vec{e_2} \mod q^{d_0}-1 = \vec{c_1}+\vec{c_2}$. If $q$ is odd, we can consider the equation modulo $2$ to get $0\cdot H +\vec{e_1}+\vec{e_2} = \vec{c_1}+\vec{c_2} \mod 2$ or $\vec{e_1}+\vec{e_2} \mod 2 = \vec{c_1}+\vec{c_2}$. If $d$ is small relative to $n$, the number of entries which are $1$ in both $\vec{e_1}$ and $\vec{e_2}$ will be very small. Hence, we can guess the positions in which $\vec{e_1}$ or $\vec{e_2}$ is $1$ from the non-zero entries in $\vec{c_1}+\vec{c_2} \mod 2$ and use the attacks in Section \ref{sec:security} to recover $\vec{m}$.

In view of the above, we modify the encoding scheme to work for our trapdoor function. Suppose that the parameters $q,n,d,t$ are fixed. Our input to our trapdoor function is a vector in $\Z_{q^{d_0}-1}^{n-t}$. Thus, each entry is an $s$-bit string, where $s = 1 + \lfloor \log_2(q^{d_0}-1)\rfloor.$ We will encode an $(n-t)$-bit plaintext message $\mathcal{P}$ into the input $\vec{m}$ for the trapdoor function $f$. The ciphertext will be the output of $f$ in $\Z_{q^{d_0}-1}^{n}$. The entire encryption and decryption processes are described as follows.

Let $\vec{m} = (m_1, \ldots, m_{n-t})$ be in $\Z_{q^{d_0}-1}^{n-t}$ and let $m_i^{(j)}$ denote the $j$-th least significant bit of $m_i$, $j=0, 1, \ldots, s-1$. Further, let $\vec{m^{(j)}} = (m_1^{(j)}, \ldots, m_{n-t}^{(j)})$. Suppose the plaintext message is $\mathcal{P} = (\mathcal{P}_1, \ldots, \mathcal{P}_{n-t})$. In the following, let \texttt{Hash} denote a cryptographic hash function from $\{0,1\}^*$ to $\{0,1\}^{n-t}$. Let $f$ be the trapdoor function with all the notations in Section \ref{sec:trapdoor function}.

{\bf Private key:}
The degree $d$ polynomial $c(x)$, the $n$ elements $\alpha_1, \ldots, \alpha_n$, the unimodular matrix $T$ and the permutation matrix $P$.

{\bf Public key:}
The parameters $n, q, d, t$ and the matrix $-G$.

{\bf Encryption:}
\begin{itemize}
\item Randomly select $n-t$ bits $\vec{z}=(z_1, \ldots, z_{n-t})$.
\item Randomly select the error string $\vec{e} = (e_1, \ldots, e_n)$ satisfying the desired properties.
\item Set $\vec{m^{(0)}} = \mathcal{P} \oplus \vec{z}$.
\item Set $\vec{m^{(1)}} = \vec{z}$.
\item Set $\vec{m^{(2)}} = \texttt{Hash}(\mathcal{P}||\vec{z}||\vec{e})$.
\item For $j=3, \ldots, s-1$, set $\vec{m^{(j)}}$ randomly.
\item Let $H=\left(\begin{matrix}I_{n-t}, & -G\end{matrix}\right)$. Then, the ciphertext $\vec{c}$ is $\vec{c} = f(\vec{m}, \vec{e}) = \vec{m}H+\vec{e}$.
\end{itemize}

{\bf Decryption:}
Given a ciphertext $\vec{c}$, the decryption proceeds as follows:
\begin{itemize}
\item Compute $\vec{m} = f^{-1}(\vec{c})$ using the private key. Let $\vec{e}$ be the corresponding error. If $\vec{e}$ contains only $0$ or $1$ entries with exactly $d-1$ $1$'s, then continue. Otherwise, decryption fails.
\item Write $\vec{m} = (m_i^{(j)})_{i=1, \ldots, n-t, j = 0, 1, \ldots, s-1}$.
\item Set $\mathcal{P}' = \vec{m^{(0)}} \oplus \vec{m^{(1)}}$.
\item If ${\rm Hash}(\mathcal{P}'||\vec{m^{(1)}}||\vec{e}) = \vec{m^{(2)}}$, then $\mathcal{P} = \mathcal{P}'$ and the decryption is successful. Otherwise, decryption fails.
\end{itemize}

\begin{remark}
\begin{itemize}
\item Like the GGH scheme \cite{GGH97}, we encode our plaintext bits in the least significant bits of the input to our trapdoor function.

\item In our scheme, the input includes not only the plaintext bits but the error bits and random bits as well. By including the error bits to the input, changing bits of the ciphertext will likely make the decryption process fail. This helps to prevent reactive attacks where attackers try to guess the error bits by sending modified ciphertexts.

\item Similar to the conversion schemes suggested in \cite{KH01} for the McEliece encryption scheme, random bits and the hash of the plaintext bits are added to the input to ensure semantic security and to prevent other attacks such as related message attacks.
\end{itemize}
\end{remark}

\subsection{Choosing the Parameters} \label{sec:parameter_section}

In view of the attacks presented in Section \ref{sec:security}, we will choose the parameters $q,n,d,t$ to resist all the possible attacks. Concretely, the following choices will be made.

\begin{itemize}
\item We let $t=d$, that is, $c(x)$ is a product of linear polynomials.
\item We set $q$ to be the smallest prime bigger than $n+d$.
\item We set $d$ to satisfy $20 \le d \le n/2$.
\item For a security level $\lambda$, we set $n$ and $d$ so that ${n -d \choose l} \ge 2^\lambda$, where $l = \frac{(n-d)(d-1)}{n}$.
\end{itemize}

With $d=t$, we have $d_0=1$ so all our operations are done modulo $q-1$. Our public key size is $(n-t)t(1+\lfloor \log_2(q-2)\rfloor)$ bits. Since encryption only involves matrix multiplication modulo $q-1$, encryption is efficient with complexity $O(n^2)$.

We now provide possible sets of values of $n$ and $d$ for our encryption scheme. For each pair of $n$ and $d$, we compute the biggest $\delta$ such that $\sqrt{\frac{n}{2\pi e(d-1)}} q^{d/n} \le \epsilon*\delta^n$, where $\epsilon=0.3$ (as explained in the appendix). We then give the corresponding approximate BKZ block size $\beta$ to achieve this $\delta$ as well as a lower bound on the estimated cost. In particular, our lower bound on the cost of the embedding attack is computed as $nE$, where $E$ is the estimated enumeration cost for one sub-routine given in \cite[Table 3, 4]{CN11} corresponding to the block size $\beta_0$, with $\beta_0$ being the largest block size smaller than $\beta$ available in the tables.

\begin{table}[ht!]
\centering
\caption{Possible $n$ and $d$ for Practical Encryption Scheme}
\begin{tabular}{|c|c|c|c|c|c|c|c|}
\hline
$n$ &$d$ & $q$ & $\log_2{n-d \choose l}$ & $\delta$ & approximate block size & $\log_2$(estimated cost) & public key size \\ \hline
230 & 29 & 263 &     106          & 1.0067    & 168   &  225.95   &52461bits=6.40KB \\
230 & 30 & 263 &     108          & 1.0067    & 168   &  225.95   &54000bits=6.59KB \\
240 & 29 & 271 &     108          & 1.0064    & 168   &  227.72   &55071bits=6.72KB \\
240 & 30 & 271 &     110          & 1.0064    & 168   &  227.72   &56700bits=6.92KB \\
240 & 31 & 277 &     113          & 1.0065    & 168   &  227.72   &58311bits=7.12KB \\
240 & 32 & 277 &     113          & 1.0065    & 168   &  227.72   &59904bits=7.31KB \\
240 & 33 & 277 &     115          & 1.0065    & 168   &  227.72   &61479bits=7.50KB \\ \hline
260 & 29 & 293 &     111          & 1.0059    & 216   &  356.19   &60291bits=7.36KB \\
260 & 30 & 293 &     114          & 1.0059    & 216   &  356.19   &62100bits=7.58KB \\
260 & 31 & 293 &     117          & 1.0059    & 216   &  356.19   &63891bits=7.80KB \\
260 & 32 & 293 &     119          & 1.0060    & 216   &  356.19   &65664bits=8.02KB \\ \hline
\end{tabular}
\label{table:estimated_cost}
\end{table}


\appendix
\section{More Experimental Results on the Embedding Attack} \label{appendix:more_embedding_attack_data}
As in Table \ref{table:embedding_attack_result}, we provide more experimental data on employing the embedding attack to our trapdoor function in Table \ref{table:more_embedding_attack_result}. In these experiments, we vary $n$ as well as the BKZ block sizes and record the $\epsilon$ that results in a successful attack after $14$ tries. From our results, there does not seem to be a discernible trend for the constant $\epsilon$. Nonetheless, we see that the minimal $\epsilon$ withstanding the embedding attack is $0.47992$. In our selection of parameters for the encryption scheme given in Table \ref{table:estimated_cost}, we therefore use $\epsilon = 0.3 < 0.47992$ to guide us in choosing the appropriate $d$ to withstand the embedding attack.

\begin{table}[ht!]
\centering
\caption{More embedding attack experimental result}
\begin{tabular}{|c|p{0.8cm}<{\centering}|p{0.6cm}<{\centering}|p{0.8cm}<{\centering}|c|c|c|}
\hline
status & $n$& $d$& $q$ & block size & $\delta$ &experiment $\epsilon$ \\ \hline
F      & 100& 27 & 223 & 30         &1.01135& 0.57134 \\
T      & 100& 28 & 223 & 30         &1.01135& 0.58862 \\ \hline
F      & 110& 33 & 149 & 30         &1.01135& 0.58050 \\
T      & 110& 34 & 149 & 30         &1.01135& 0.59820 \\ \hline
F      & 120& 37 & 163 & 30         &1.01135& 0.54738 \\
T      & 120& 38 & 163 & 30         &1.01135& 0.56332 \\ \hline
F      & 130& 44 & 179 & 30         &1.01135& 0.56039 \\
T      & 130& 45 & 179 & 30         &1.01135& 0.57651 \\ \hline
F      & 140& 52 & 197 & 30         &1.01135& 0.58647 \\
T      & 140& 53 & 197 & 30         &1.01135& 0.60312 \\ \hline
F      & 150& 57 & 211 & 30         &1.01135& 0.55582 \\
T      & 150& 58 & 211 & 30         &1.01135& 0.57092 \\ \hline
F      & 160& 66 & 229 & 30         &1.01135& 0.58584 \\
T      & 160& 67 & 229 & 30         &1.01135& 0.60145 \\ \hline\hline
F      & 200&106 & 311 & 20         &1.01168& 0.68466 \\
T      & 200&107 & 311 & 20         &1.01168& 0.70125 \\ \hline
F      & 200& 95 & 307 & 30         &1.01135& 0.55993 \\
T      & 200& 96 & 307 & 30         &1.01135& 0.57315 \\ \hline
F      & 200& 94 & 307 & 40         &1.01119& 0.56464 \\
T      & 200& 95 & 307 & 40         &1.01119& 0.57793 \\ \hline
F      & 200& 86 & 293 & 50         &1.01098& 0.47992 \\
T      & 200& 87 & 293 & 50         &1.01098& 0.47992 \\ \hline
F      & 200& 80 & 283 & 55         &1.01007& 0.49576 \\
T      & 200& 81 & 283 & 55         &1.01007& 0.50674 \\ \hline
F      & 200& 79 & 281 & 60         &1.00987& 0.50323 \\
T      & 200& 80 & 283 & 60         &1.00987& 0.51579 \\ \hline
F      & 200& 77 & 281 & 65         &1.00934& 0.53522 \\
T      & 200& 78 & 281 & 65         &1.00934& 0.54692 \\ \hline\hline
F      & 220&123 & 347 & 20         &1.01168& 0.66355 \\
T      & 220&124 & 347 & 20         &1.01168& 0.67864 \\ \hline
F      & 220&115 & 337 & 30         &1.01135& 0.58720 \\
T      & 220&116 & 347 & 30         &1.01135& 0.60966 \\ \hline
F      & 220&113 & 337 & 40         &1.01119& 0.58180 \\
T      & 220&114 & 337 & 40         &1.01119& 0.59474 \\ \hline
F      & 220&111 & 337 & 50         &1.01098& 0.58286 \\
T      & 220&112 & 337 & 50         &1.01098& 0.59578 \\ \hline
F      & 220&109 & 331 & 55         &1.01007& 0.67414 \\
T      & 220&110 & 337 & 55         &1.01007& 0.69519 \\ \hline
F      & 220&107 & 331 & 60         &1.00987& 0.67427 \\
T      & 220&108 & 331 & 60         &1.00987& 0.68903 \\ \hline\hline
\end{tabular}

\label{table:more_embedding_attack_result}
\end{table}

\bibliographystyle{plain}
\bibliography{QuantumSafeCrypto}

\end{document}